\documentclass[12pt]{article}

\usepackage[T1]{fontenc}
\usepackage[utf8]{inputenc}
\usepackage{lmodern}
\usepackage{microtype}
\usepackage[letterpaper,margin=1in]{geometry}
\usepackage{amsmath,amssymb,amsthm,mathtools}
\usepackage{amsmath,amssymb,amsthm,mathtools}
\usepackage{tikz}
\usetikzlibrary{arrows.meta}
\usepackage[hidelinks]{hyperref}
\tikzset{
  dvertex/.style={
    circle,
    draw,
    fill=white,
    minimum size=7mm,
    inner sep=0pt,
    font=\small
  },
  cyclearc/.style={
    -{Stealth[length=2.2mm,width=1.5mm]},
    very thick,
    shorten <=2.5pt,
    shorten >=2.5pt
  },
  antarc/.style={
    -{Stealth[length=2.0mm,width=1.4mm]},
    thick,
    dashed,
    shorten <=2.5pt,
    shorten >=2.5pt
  }
}
\usepackage[hidelinks]{hyperref}

\newtheorem{theorem}{Theorem}
\newtheorem{lemma}[theorem]{Lemma}

\newtheorem{problem}[theorem]{Problem}
\theoremstyle{remark}

\newcommand{\UG}{\operatorname{UG}}
\newcommand{\cB}{\mathcal{B}}
\newcommand{\NP}{\mathsf{NP}}

\title{The Complexity of Mixed Arc-Disjoint Spanning Subdigraphs with Antistrong Connectivity}
\author{
Jiangdong Ai\thanks{School of Mathematical Sciences and LPMC,
Nankai University, Tianjin 300071, P.R. China.
Email: {\tt jd@nankai.edu.cn}.}
\and
Gregory Gutin\thanks{Corresponding author. Department of Computing, Security and Mathematics,
Royal Holloway, University of London, Egham, Surrey TW20 0EX, UK.
Email: {\tt g.gutin@rhul.ac.uk}.}
\and
Hui Lei\thanks{School of Statistics and Data Science, LPMC and KLMDASR,
Nankai University, Tianjin 300071, P.R. China.
Email: {\tt hlei@nankai.edu.cn}.}
\and
Yongtang Shi\thanks{Center for Combinatorics and LPMC,
Nankai University, Tianjin 300071, P.R. China.
Email: {\tt shi@nankai.edu.cn}.}}
\date{}

\begin{document}
\maketitle

\begin{abstract}
A trail is antidirected if its arcs alternate between forward and backward. A digraph $D$  is antistrong if, for every ordered pair of distinct vertices $x,y\in V(D)$, it contains a forward antidirected $(x,y)$-trail. Bang-Jensen, Bessy, Jackson and Kriesell
[J. Combin. Theory Ser. B 122 (2017), 68--90]
introduced antistrong connectivity and posed two problems concerning
mixed arc-disjoint spanning subdigraphs.  In the first problem, one
seeks an antistrong spanning subdigraph and an arc-disjoint strong
spanning subdigraph.  In the second, strong connectivity is replaced
by the requirement that the underlying graph of the second
subdigraph be 2-edge-connected.  Bang-Jensen et al.\ asked whether
each of the two problems can be solved in polynomial time.

We prove that the two associated decision problems are NP-complete.
The first remains NP-complete for digraphs with maximum out-degree
at most four and maximum in-degree at most five.  The second remains
NP-complete even for oriented digraphs that are strong and
antistrong, whose underlying graphs are 3-vertex-connected, and in
which all but at most two vertices have both in-degree and out-degree
at most four.  In particular, the latter hardness result does not
rely on digons.
\end{abstract}


\section{Introduction}

The packing of edge-disjoint spanning structures is a central theme in
graph theory and combinatorial optimization.  In undirected graphs, the
classical theorems of Nash-Williams and Tutte characterize the existence
of a prescribed number of edge-disjoint spanning trees
\cite{NashWilliams,Tutte}.  Edmonds' branching theorem provides a
fundamental directed counterpart for arc-disjoint rooted branchings
\cite{Edmonds}; see also the monograph of Bang-Jensen and
Gutin~\cite{BJGbook} for background on directed connectivity and
branchings.

A different line of research asks for decompositions into strong
spanning subdigraphs.  Such questions, often called \emph{strong arc
decomposition} problems, have been studied both from structural and
algorithmic viewpoints; see the survey of Bang-Jensen and
Kriesell~\cite{BJKsurvey}.  Although the problem is computationally
hard for general digraphs~\cite{BY}, precise positive results are known
for several structured classes, including tournaments and semicomplete
compositions~\cite{BJGY,BJY}, as well as broader classes of compositions
and products~\cite{SGA}.  These results concern homogeneous packings,
where the packed spanning subdigraphs satisfy the same connectivity
condition.  The present paper deals instead with \emph{mixed} packings,
in which the two spanning requirements are different.

Let $T=v_0e_1v_1e_2\cdots e_kv_k$ be a trail in a digraph.  An arc $e_i$ is \emph{forward} in $T$
if $e_i=v_{i-1}v_i$, and \emph{backward} if
$e_i=v_iv_{i-1}$.  The trail $T$ is \emph{antidirected} if its arcs alternate between forward and backward.  It is a
\emph{forward antidirected trail} if its first and last arcs are forward; equivalently, $k$ is odd and
\[
   e_i=
   \begin{cases}
      v_{i-1}v_i, & \text{if $i$ is odd},\\
      v_iv_{i-1}, & \text{if $i$ is even}.
   \end{cases}
\]

A digraph $D$ of order at least three is \emph{antistrong} if, for every ordered pair of distinct vertices $x,y$, it contains a forward antidirected $(x,y)$-trail.  This notion was introduced and
systematically studied by Bang-Jensen, Bessy, Jackson and Kriesell~\cite{BBJK}.  A central observation of their work is that antistrong connectivity is ordinary connectivity in a natural
bipartite representation.  Consequently, antistrong connectivity can be recognized in linear time, and several augmentation and
packing problems admit matroidal formulations.

In particular, antistrong spanning connectivity is governed by the graphic matroid of the bipartite representation, making polynomial and matroidal approaches plausible when it is considered in isolation.  However, these favourable properties do not extend to mixed packing settings where an antistrong spanning subdigraph must coexist with a spanning subdigraph satisfying a different connectivity requirement. The results below show that coupling it with a different
spanning-connectivity requirement nevertheless leads to NP-complete mixed packing problems.

A digraph $D$ is \emph{strong} if, for every ordered pair of distinct vertices $x,y\in V(D)$, it contains a directed
$(x,y)$-path. Here and throughout, $\UG(D)$ is obtained by suppressing the
orientations of the arcs of $D$ and then replacing all parallel
edges by a single edge.  Thus $\UG(D)$ is always a simple graph.
An edge of a graph is a \emph{bridge} if its deletion increases the number of connected components.  A graph is \emph{2-edge-connected} if it is connected and has no bridge.
A graph with at least four vertices is \emph{3-vertex-connected} if deleting any set of at most two vertices leaves a connected graph. A pair of opposite arcs forms a \emph{digon}.  A digraph is
\emph{oriented} if it contains no digon.
For a vertex $v$ of a digraph $D$, let $d_D^+(v)$ and $d_D^-(v)$ denote its out-degree and in-degree, respectively, and let $\Delta^+(D)$ and $\Delta^-(D)$ denote the corresponding maximum
degrees.
Two subdigraphs are \emph{arc-disjoint} if their arc
sets are disjoint.

Bang-Jensen et al.~\cite[Questions~9.1 and~9.2]{BBJK} asked the following two questions.

\begin{problem}[Antistrong--strong spanning packing]
\label{prob:AS-SC}
Can we decide in polynomial time whether a digraph $D$ contains arc-disjoint spanning
subdigraphs $D_1,D_2$ such that $D_1$ is antistrong and $D_2$ is strongly connected?
\end{problem}

\begin{problem}[Antistrong--2-edge-connected spanning packing]
\label{prob:AS-2EC}
Can we decide in polynomial time whether a digraph $D$ contains arc-disjoint spanning
subdigraphs $D_1,D_2$ such that $D_1$ is antistrong and
$\UG(D_2)$ is 2-edge-connected?
\end{problem}

Our main results show that the corresponding decision problems are NP-complete.  Consequently, both questions have negative answers
unless $\mathsf{P}=\mathsf{NP}$.

\begin{theorem}\label{thm:main-SC}
It is NP-complete to decide whether a digraph $D$ contains
arc-disjoint spanning subdigraphs $D_1,D_2$ such that $D_1$ is antistrong and $D_2$ is strong. This remains NP-complete even when
$\Delta^+(D)\le 4$ and $\Delta^-(D)\le 5$.
\end{theorem}

\begin{theorem}\label{thm:main-2EC}
It is NP-complete to decide whether a digraph $D$ contains
arc-disjoint spanning subdigraphs $D_1,D_2$ such that $D_1$ is antistrong and $\UG(D_2)$ is 2-edge-connected. It remains NP-complete even for oriented digraphs $D$ that are strong and antistrong, whose underlying graphs $\UG(D)$ are
3-vertex-connected, and in which all but at most two vertices $v\in V(D)$ satisfy
$d^+(v)\leq4$ and $d^-(v)\leq4$.
\end{theorem}

For Problem~\ref{prob:AS-SC}, we reduce from the non-separating strong spanning subdigraph problem, which is NP-complete already for 2-regular digraphs~\cite{BY}.  The reduction uses a four-vertex switch containing two arc-disjoint spanning subdigraphs, one antistrong and the other strong.

For Problem~\ref{prob:AS-2EC}, we give a polynomial reduction from Hamiltonian Path in planar cubic bipartite graphs, which is
NP-complete by Munaro~\cite[Theorem~23]{Munaro}. This restriction
belongs to the classical line of planar Hamiltonicity hardness, including the seminal result of Garey, Johnson and Tarjan~\cite{GJT}.  The second reduction uses a seven-vertex oriented core containing an antistrong spanning subdigraph and an
arc-disjoint spanning directed $7$-cycle.  Orienting every edge of
the bipartite source graph $G=(U,W;E)$ from $U$ to $W$ yields an oriented input digraph. The same construction has a
3-vertex-connected underlying graph, and only two of its vertices can have unbounded in-degree or out-degree.

The paper is organized as follows. In Section~\ref{sec:prelim} we collect the common preliminaries, including the bipartite representation of a digraph and an elementary observation about quotients.  In Section~\ref{sec:strong} we introduce the four-vertex switch and prove Theorem~\ref{thm:main-SC}. In
Section~\ref{sec:2ec} we introduce the seven-vertex oriented core and prove Theorem~\ref{thm:main-2EC}.  In the final section, we present some concluding remarks.

\section{Preliminaries}\label{sec:prelim}

All digraphs in this paper are finite and loopless.  Opposite arcs are allowed, but there is at most one arc with any prescribed ordered pair of ends.  All undirected graphs are finite and simple.
A subdigraph is \emph{spanning} if it has the same vertex set as the ambient digraph.

For distinct vertices $u,v$, we write $uv$ for the arc directed from $u$ to $v$.  For $X\subseteq A(D)$, we write $D[X]=(V(D),X)$ and $D-X=(V(D),A(D)\setminus X)$.
For a graph $G$ and $F\subseteq E(G)$, let $d_F(v)$ denote the degree of $v$ in the spanning subgraph $(V(G),F)$. For $Z\subseteq V(G)$, let $\delta_G(Z)$ denote the set of edges of $G$ with exactly one endpoint in $Z$, and let $G-Z$ denote the graph obtained by deleting the vertices in $Z$ (and all incident edges).

For a digraph $D=(V,A)$, its \emph{bipartite representation} is the
bipartite graph $\cB(D)=(V'\cup V'',E)$ with $E=\{u'v'':uv\in A\}$ where $V'=\{v':v\in V\}$ and $V''=\{v'':v\in V\}$ are disjoint copies of $V$. We refer to $v'$ and $v''$ as the two \emph{clones} of $v$ in
$\cB(D)$.  We use the following characterization of
Bang-Jensen, Bessy, Jackson and Kriesell~\cite{BBJK}.

\begin{theorem}[Bang-Jensen--Bessy--Jackson--Kriesell]
\label{thm:bipartite}
Let $D$ be a digraph of order at least three.  Then $D$ is antistrong if and only if $\cB(D)$ is connected.
\end{theorem}

For a partition $\mathcal P$ of the vertex set of a digraph $D$, the quotient $D/\mathcal P$ is obtained by identifying every part of $\mathcal P$ to one vertex, deleting loops, and suppressing parallel arcs with the same orientation.  For an undirected graph $G$, the quotient $G/\mathcal P$ is defined similarly, with parallel edges replaced by a single edge. These suppressions do not affect
strong connectivity or connectivity.

\begin{lemma}\label{lem:quotient}
Let $\mathcal P$ be a partition of the relevant vertex set.
\begin{enumerate}
\renewcommand{\labelenumi}{\textup{(\roman{enumi})}}
\item If a digraph $D$ is strong, then $D/\mathcal P$ is strong.
\item If a graph $G$ is connected, then $G/\mathcal P$ is connected.
\end{enumerate}
\end{lemma}

\begin{proof}
A directed path in $D$ projects to a directed walk in
$D/\mathcal P$. Deleting loops and suppressing parallel arcs do not destroy this walk, proving (i). The same argument with an undirected path proves (ii).
\end{proof}

\section{Packing an antistrong and a strong spanning subdigraph}\label{sec:strong}

We use the following NP-complete problem of Bang-Jensen and
Yeo~\cite[Theorem~1.6]{BY}.

\begin{problem}\label{prob:NSSS}
The input is a 2-regular digraph $H$, meaning that every vertex has
in-degree two and out-degree two.  Decide whether $H$ has a spanning
strong subdigraph $S$ such that $\UG(H-A(S))$ is connected.
\end{problem}
We next define the switch used in the first reduction.  Let $Q$ be the
digraph on vertex set $\{0,1,2,3\}$ with two disjoint arc sets $C=\{01,12,23,30\}$ and $F= \{10,21,32,02,20,13,31\}$.
We call $0$ the \emph{port} of $Q$. Whenever $Q$ is used inside a larger construction, every arc entering or leaving the copy is incident with its port. The switch is shown in
Figure~\ref{fig:switch-Q}.

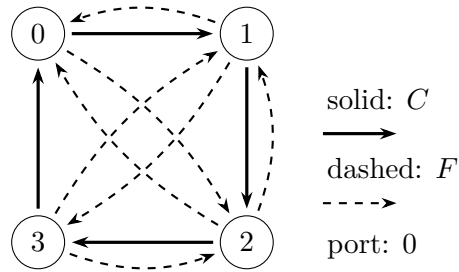
\begin{figure}[htbp]
\centering
\begin{tikzpicture}[scale=1.15]
  \node[dvertex] (q0) at (0,1.4) {$0$};
  \node[dvertex] (q1) at (2.4,1.4) {$1$};
  \node[dvertex] (q2) at (2.4,-1.0) {$2$};
  \node[dvertex] (q3) at (0,-1.0) {$3$};

  \draw[cyclearc] (q0) -- (q1);
  \draw[cyclearc] (q1) -- (q2);
  \draw[cyclearc] (q2) -- (q3);
  \draw[cyclearc] (q3) -- (q0);

  \draw[antarc] (q1) to[bend right=20] (q0);
  \draw[antarc] (q2) to[bend right=20] (q1);
  \draw[antarc] (q3) to[bend right=20] (q2);
  \draw[antarc] (q0) to[bend left=14] (q2);
  \draw[antarc] (q2) to[bend left=14] (q0);
  \draw[antarc] (q1) to[bend left=14] (q3);
  \draw[antarc] (q3) to[bend left=14] (q1);

  \node[font=\small,anchor=west] at (3.2,0.65) {solid: $C$};
  \draw[cyclearc] (3.2,0.25) -- (4.2,0.25);

  \node[font=\small,anchor=west] at (3.2,-0.15) {dashed: $F$};
  \draw[antarc] (3.2,-0.55) -- (4.2,-0.55);

  \node[font=\small,anchor=west] at (3.2,-1.05) {port: $0$};
\end{tikzpicture}
\caption{The four-vertex switch $Q$.  Solid arcs form $C$, and
dashed arcs form $F$; vertex $0$ is the port.}
\label{fig:switch-Q}
\end{figure}

\begin{lemma}\label{lem:switch}
The following statements hold.
\begin{enumerate}
\renewcommand{\labelenumi}{\textup{(\roman{enumi})}}
\item $Q[C]$ is a directed 4-cycle.  In particular, it is strong and
$\UG(Q[C])$ is 2-edge-connected.
\item $Q[F]$ is antistrong.
\item $C\cap F=\emptyset$, and $\UG(Q[C\cup F])\cong K_4$.
\end{enumerate}
\end{lemma}

\begin{proof}
We only need to verify (ii). The bipartite representation of $Q[F]$ is the path
$0'-2''-3'-1''-2'-0''-1'-3''$.
It is therefore connected, so Theorem~\ref{thm:bipartite} implies that $Q[F]$ is antistrong.  The remaining assertions follow directly from the displayed arc sets.
\end{proof}

We write $Q_z$ for a copy of $Q$ indexed by a symbol $z$, and also
write $z$ for its port when no ambiguity can arise.  The corresponding
copies of $C$ and $F$ are denoted by $C_z$ and $F_z$.

\begin{proof}[Proof of Theorem~\ref{thm:main-SC}]
Membership in $\NP$ is immediate.  Given two arc sets, one checks
arc-disjointness, strong connectivity of the second subdigraph, and
connectivity of the bipartite representation of the first subdigraph
in polynomial time.

We reduce from Problem~\ref{prob:NSSS}.  Let $H=(V,A)$ be a 2-regular
digraph.  Replace every vertex $v\in V$ by a copy $Q_v$ of the switch.
For every arc $uv\in A$, add the arc from the port of $Q_u$ to the port
of $Q_v$.  Denote the resulting digraph by $\widehat H$.  We call these
new arcs \emph{linking arcs}.  The construction is clearly polynomial.

Suppose first that $H$ contains a spanning strong subdigraph $S$ such that $\UG(H-A(S))$ is connected. Define two spanning subdigraphs of $\widehat H$ by
\begin{align*}
 A(D_{\rm str})
   &=
   \bigcup_{v\in V} C_v
   \ \cup\
   \{uv:uv\in A(S)\},\\
 A(D_{\rm anti})
   &=
   \bigcup_{v\in V} F_v
   \ \cup\
   \{uv:uv\in A\setminus A(S)\},
\end{align*}
where in the two displayed sets $uv$ denotes the linking arc between the corresponding ports. These subdigraphs are
arc-disjoint by Lemma~\ref{lem:switch}.

Every $Q_v[C_v]$ is strong.  After contracting the switch copies, the quotient of $D_{\rm str}$ is precisely $S$. Equivalently, directed
paths in $S$ can be lifted through the strong switch copies, and hence $D_{\rm str}$ is strong.

For every $v$, the graph $\cB(Q_v[F_v])$ is connected. The linking arcs of $D_{\rm anti}$ join these connected bipartite blocks according
to the connected graph $\UG(H-A(S))$.  Hence $\cB(D_{\rm anti})$ is connected, and $D_{\rm anti}$ is antistrong by Theorem~\ref{thm:bipartite}.  Thus $\widehat H$ contains the required pair of arc-disjoint spanning subdigraphs.

Conversely, suppose that $\widehat H$ contains arc-disjoint spanning
subdigraphs $D_{\rm anti},D_{\rm str}$ such that $D_{\rm anti}$ is antistrong and $D_{\rm str}$ is strong.  Let
\begin{align*}
 X&=\{uv\in A(H):\text{the linking arc $uv$ belongs to }D_{\rm str}\},\\
 Y&=\{uv\in A(H):\text{the linking arc $uv$ belongs to }D_{\rm anti}\}.
\end{align*}
Contract every copy $Q_v$ in $D_{\rm str}$.  By
Lemma~\ref{lem:quotient}(i), the resulting quotient is strong.  Its non-loop arcs are exactly the linking arcs corresponding to $X$.
Consequently, the spanning subdigraph $H[X]$ is strong.

Since $D_{\rm anti}$ is antistrong, $\cB(D_{\rm anti})$ is connected.
In this bipartite graph identify, for every $v\in V$, all eight bipartite clones corresponding to the four vertices of $Q_v$.
By Lemma~\ref{lem:quotient}(ii), the quotient is connected.  Its non-loop edges are exactly the edges represented by the linking arcs corresponding to $Y$.  Therefore $\UG(H[Y])$ is connected.

The two packed subdigraphs are arc-disjoint, so $X\cap Y=\emptyset$.
Hence $\UG(H-X)$ contains the connected spanning subgraph
$\UG(H[Y])$ and is connected.  Thus $H[X]$ is a spanning strong subdigraph whose deletion leaves a connected underlying graph, and
$H$ is a yes-instance of Problem~\ref{prob:NSSS}.  This proves NP-hardness.

It remains to verify the degree promise. The port $0$ has two outgoing and three incoming arcs inside $Q$, while every other switch
vertex has both in-degree and out-degree at most three.  Since $H$ is 2-regular, at most two linking arcs enter and at most two linking arcs
leave each port.  Thus
$\Delta^+(\widehat H)\le 4$
and $\Delta^-(\widehat H)\le 5$.
\end{proof}

\section{Packing an antistrong and a 2-edge-connected spanning subdigraph}\label{sec:2ec}

We now turn to Problem~\ref{prob:AS-2EC}.  The second reduction is independent of the four-vertex switch used in Section~\ref{sec:strong}.  Apart from the bipartite
characterization of antistrong connectivity, it uses a different construction.  Its local ingredient is a seven-vertex oriented core, which allows us to avoid opposite arcs altogether.

Let $R$ be the digraph on vertex set $\{0,1,\ldots,6\}$ whose arc set is the disjoint union of
\begin{align*}
 C_7&=\{01,12,23,34,45,56,60\},\\
 F_7&=\{30,04,50,13,41,15,61,24,52,26,35,36,46\}.
\end{align*}
The core is shown in Figure~\ref{fig:oriented-core-R}.

\begin{figure}[htbp]
\centering
\begin{tikzpicture}
  \foreach \i/\ang in
    {0/90,1/38.571,2/-12.857,3/-64.286,
     4/-115.714,5/-167.143,6/141.429}
    \node[dvertex] (r\i) at (\ang:3.25cm) {$\i$};

  \draw[cyclearc] (r0) -- (r1);
  \draw[cyclearc] (r1) -- (r2);
  \draw[cyclearc] (r2) -- (r3);
  \draw[cyclearc] (r3) -- (r4);
  \draw[cyclearc] (r4) -- (r5);
  \draw[cyclearc] (r5) -- (r6);
  \draw[cyclearc] (r6) -- (r0);

  \draw[antarc] (r3) to[bend left=0] (r0);
  \draw[antarc] (r0) to[bend left=2] (r4);
  \draw[antarc] (r5) to[bend right=10] (r0);
  \draw[antarc] (r1) to[bend right=10] (r3);
  \draw[antarc] (r4) to[bend left=2] (r1);
  \draw[antarc] (r1) to[bend left=2] (r5);
  \draw[antarc] (r6) to[bend right=10] (r1);
  \draw[antarc] (r2) to[bend right=10] (r4);
  \draw[antarc] (r5) to[bend left=2] (r2);
  \draw[antarc] (r2) to[bend left=2] (r6);
  \draw[antarc] (r3) to[bend right=10] (r5);
  \draw[antarc] (r3) to[bend left=0] (r6);
  \draw[antarc] (r4) to[bend right=10] (r6);

  \node[font=\small,anchor=west] at (4.2,0.65)
    {solid: $C_7$};
  \draw[cyclearc] (4.2,0.25) -- (5.25,0.25);

  \node[font=\small,anchor=west] at (4.2,-0.15)
    {dashed: $F_7$};
  \draw[antarc] (4.2,-0.55) -- (5.25,-0.55);
\end{tikzpicture}
\caption{The seven-vertex oriented core $R$.  Solid arcs form the
directed cycle $C_7$, and dashed arcs form $F_7$.}
\label{fig:oriented-core-R}
\end{figure}
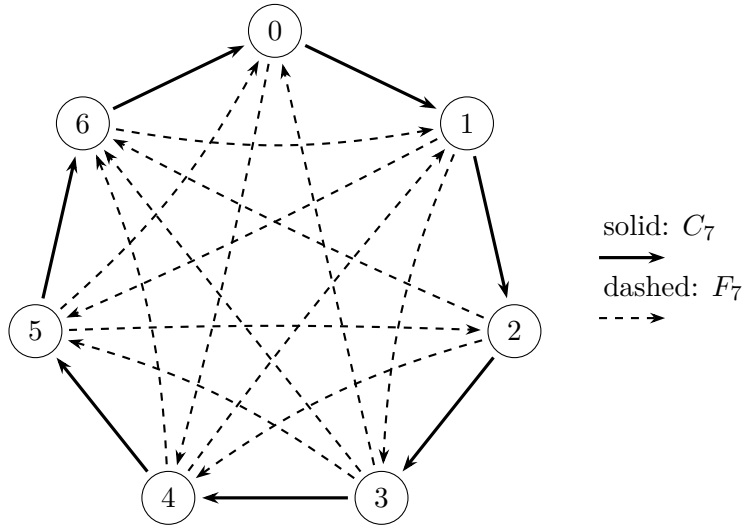

\begin{lemma}\label{lem:oriented-core}
The digraph $R$ is oriented, $R[C_7]$ is a directed $7$-cycle, and $R[F_7]$ is antistrong.
\end{lemma}

\begin{proof}
No unordered pair of vertices occurs in both $C_7$ and $F_7$, and no two arcs in their union are opposite.  Thus $R$ is oriented, while the assertion about $C_7$ is immediate.

It remains to verify that $R[F_7]$ is antistrong. In the bipartite representation $\cB(R[F_7])$, consider the four paths
\begin{align*}
 &3'-0''-5'-2'',\\
 &3'-5''-1'-3'',\\
 &3'-6''-2'-4''-0',\\
 &6''-4'-1''-6'.
\end{align*}
The first three paths meet at $3'$, while the fourth meets the third at $6''$. Hence their union is connected. It contains all fourteen vertices of $\cB(R[F_7])$ and exactly thirteen edges, namely the edges represented by the arcs of $F_7$. Thus this union is a tree. In particular, $\cB(R[F_7])$ is connected, and Theorem~\ref{thm:bipartite} implies that $R[F_7]$ is antistrong.
\end{proof}

We shall also use the following immediate observation:
\begin{equation}\label{eq:underlying-core}
   \UG(R)\cong K_7-02,
\end{equation}
since the twenty unordered pairs represented by $C_7\cup F_7$ are all distinct and comprise all pairs of core vertices except $\{0,2\}$.

We use the following NP-complete problem. Here a cubic graph is a 3-regular graph.

\begin{problem}
\label{prob:PCBHAM}
The input is a planar cubic bipartite graph $G$. Decide whether $G$ contains a Hamiltonian path.
\end{problem}

Problem~\ref{prob:PCBHAM} is NP-complete by a theorem of
Munaro~\cite[Theorem~23]{Munaro}.  Note that every yes-instance is connected.  Hence, when reducing from
Problem~\ref{prob:PCBHAM}, a disconnected source instance can be recognized in polynomial time and sent to a fixed no-instance satisfying all the promises in Theorem~\ref{thm:main-2EC}.

For completeness, let $D_0$ be the tournament on
$\{0,1,2,3,4\}$ with
\[
 A(D_0)=\{10,02,30,04,21,31,41,32,42,43\}.
\]
It is oriented because it is a tournament, and it is strong because
it contains the spanning directed cycle $0\to 4\to 3\to 2\to 1\to 0$.
Its bipartite representation contains a spanning tree represented by the nine arcs
$02,04,32,42,30,31,43,10,21$,
so $D_0$ is antistrong.  Moreover, $\UG(D_0)=K_5$.
which is 3-vertex-connected, and every vertex of $D_0$ has both in-degree and out-degree at most four.

It remains to verify that $D_0$ is a no-instance.  Any antistrong spanning subdigraph $J_1$ of $D_0$ has a connected bipartite representation on ten vertices, and hence $|A(J_1)|\ge 9$.

If $J_2$ is a spanning subdigraph whose underlying graph is 2-edge-connected, then $\delta(\UG(J_2))\ge 2$,
and the handshaking lemma gives $|A(J_2)|=|E(\UG(J_2))|\ge 5$.
Thus two such arc-disjoint spanning subdigraphs would require at
least $9+5=14$ distinct arcs, whereas $D_0$ has only ten arcs.
Therefore $D_0$ is a no-instance.  We may henceforth assume that the
source graph is connected.

Let $G=(U,W;E)$ be a connected planar cubic bipartite graph, and choose
one vertex $u_*\in U$.  Construct a digraph $D(G)$ as follows. Start with one copy of $R$, and add one vertex $x_v$, called a \emph{terminal}, for every $v\in V(G)$. Add the arcs $0x_u,\ 1x_u$ for every $u\in U$;
$x_w0,\ x_w1$ for every $w\in W$; $x_ux_w$ for every $uw\in E$ and $x_{u_*}2$.

The arcs $x_ux_w$ are called \emph{variable arcs}, the arcs
incident with $0$ or $1$ and a terminal $x_v$ are called
\emph{spokes}, and $x_{u_*}2$ is called the \emph{connector}.  All variable arcs are directed from the $U$-side to the $W$-side. By Lemma~\ref{lem:oriented-core}, and because the connector uses the core vertex $2$ whereas the spokes use only $0$ and $1$, the
digraph $D(G)$ is oriented.

The construction has $|V(G)|+7$ vertices and
$|E(G)|+2|V(G)|+21$ arcs, and can therefore be carried out in polynomial time.

\begin{proof}[Proof of Theorem~\ref{thm:main-2EC}]
The problem belongs to $\NP$: antistrong connectivity is checked by connectivity of the bipartite representation, and 2-edge-connectivity of the underlying graph is checked by a bridge
search.

We prove NP-hardness by a polynomial reduction from
Problem~\ref{prob:PCBHAM}.  Let $G=(U,W;E)$ be a connected instance,
and let $D(G)$ be the oriented graph constructed above.  We claim that $G$ has a Hamiltonian path if and only if $D(G)$ is a yes-instance of the decision problem described in
Problem~\ref{prob:AS-2EC}.
Suppose first that $P$ is a Hamiltonian path of $G$, and let $S=E\setminus E(P)$.  Define two spanning subdigraphs of $D(G)$ by
\begin{align*}
 A(D_{\rm anti})={}&F_7
   \cup\{1x_u:u\in U\}
   \cup\{x_w1:w\in W\}
   \cup\{x_{u_*}2\}\\
   &\cup\{x_ux_w:uw\in E(P)\},\\
 A(D_2)={}&C_7
   \cup\{0x_u:u\in U\}
   \cup\{x_w0:w\in W\}
   \cup\{x_ux_w:uw\in S\}.
\end{align*}
The two arc sets are disjoint.

Consider in $\cB(D_{\rm anti})$ the set $L=\{x_u':u\in U\}\cup\{x_w'':w\in W\}$.
The variable arcs corresponding to $E(P)$ induce on $L$ a graph
isomorphic to the Hamiltonian path $P$.  The connector is represented
by the edge $x_{u_*}'2''$ and therefore joins this path to the
connected graph $\cB(R[F_7])$.  For every $u\in U$, the spoke $1x_u$
attaches the remaining clone $x_u''$ to the core, and for every
$w\in W$, the spoke $x_w1$ attaches the remaining clone $x_w'$ to the
core.  Thus $\cB(D_{\rm anti})$ is connected, and
$D_{\rm anti}$ is antistrong.

Since $G$ is cubic, every internal vertex of $P$ is incident with one
edge of $S$, and each end of $P$ is incident with two edges of $S$.
In particular,
\begin{equation}\label{eq:S-meets-forward}
   d_S(v)\ge 1\qquad\text{for every }v\in V(G).
\end{equation}
The graph $\UG(D_2)$ contains the core cycle $\UG(R[C_7])$.  Every terminal $x_v$ has a spoke to $0$.  For each $uw\in S$, the variable
edge $x_ux_w$ and the two corresponding spokes form the triangle $0x_ux_w0$.
Every variable edge lies on such a triangle, and by
\eqref{eq:S-meets-forward}, every spoke lies on at least one such
triangle.  The core edges lie on the core $7$-cycle.  Hence
$\UG(D_2)$ is connected and every one of its edges lies on a cycle.
It is therefore 2-edge-connected.

Conversely, suppose that $D(G)$ contains arc-disjoint spanning
subdigraphs $D_{\rm anti}$ and $D_2$, where $D_{\rm anti}$ is
antistrong and $\UG(D_2)$ is 2-edge-connected.  Let
\begin{align*}
 T&=\{uw\in E:x_ux_w\in A(D_{\rm anti})\},\\
 S&=\{uw\in E:x_ux_w\in A(D_2)\}.
\end{align*}
Clearly $S\cap T=\emptyset$.

We first show that $(V(G),T)$ is connected. In the full bipartite representation $\cB(D(G))$, let
\[
   L=\{x_u':u\in U\}\cup\{x_w'':w\in W\}.
\]
By construction,
\begin{equation}\label{eq:unique-connector-cut}
   \delta_{\cB(D(G))}(L)=\{x_{u_*}'2''\}.
\end{equation}
Indeed, every variable edge has both endpoints in $L$, every spoke is incident with one of the other terminal clones, and the connector is the only remaining edge incident with $L$.

Since $\cB(D_{\rm anti})$ is connected,
\eqref{eq:unique-connector-cut} forces the connector to belong to $D_{\rm anti}$.  Moreover, $\cB(D_{\rm anti})[L]$ must be connected.  Otherwise, a component not containing $x_{u_*}'$ would
have no edge to the rest of $\cB(D_{\rm anti})$. The only ambient edges with both endpoints in $L$ are the variable edges. Consequently, the graph $(V(G),T)$ is connected and spanning.

We next show that
\begin{equation}\label{eq:S-positive-oriented}
   d_S(v)\ge 1\qquad\text{for every }v\in V(G).
\end{equation}
For $u\in U$, the vertex $x_u''$ of the bipartite representation is
incident in the whole ambient graph only with the two spoke edges $0'x_u''$ and $1'x_u''$.  Connectivity of $\cB(D_{\rm anti})$ forces at least one of the corresponding spokes to
belong to $D_{\rm anti}$.  Similarly, for $w\in W$, the clone $x_w'$ is incident only with the two spoke edges $x_w'0''$ and $x_w'1''$, so again $D_{\rm anti}$ uses at least one spoke at $x_w$.

By arc-disjointness, $D_2$ therefore uses at most one of the two spokes at every terminal. Apart from its two spokes, every terminal is incident only with variable arcs, except that $x_{u_*}$ is also incident with the connector. The connector is unavailable to $D_2$, because it has already been forced into $D_{\rm anti}$.

Since a 2-edge-connected graph has minimum degree at least two, every terminal must therefore be incident in $D_2$ with at least one variable edge. This proves
\eqref{eq:S-positive-oriented}.

Finally, cubicity and arc-disjointness give $d_T(v)+d_S(v)\le 3$ for every $v\in V(G)$.
Together with \eqref{eq:S-positive-oriented}, this yields $d_T(v)\le 2$ for every $v$. The graph $(V(G),T)$ is connected and spanning, so it is a Hamiltonian path or a Hamiltonian cycle.
In the latter case, deleting one edge produces a Hamiltonian path. Hence $G$ has a Hamiltonian path, completing the equivalence.

It remains to verify the additional promises. The directed cycle $R[C_7]$ makes the core strong. The core reaches every $x_u$ directly through $0x_u$, and reaches every $x_w$ through a two-arc
path $0x_ux_w$, where $u$ is any neighbour of $w$.  Conversely, every $x_w$ reaches the core through $x_w0$, and every $x_u$ reaches the core through
$x_ux_w0$, where $w$ is any neighbour of $u$. Thus $D(G)$ is strong.

The graph $\cB(D(G))$ is connected as well. Indeed, use the connected core $\cB(R[F_7])$, all spokes incident with $1$, the connector, and all variable arcs. Since $G$ is connected, the variable edges connect all vertices of $L$, the connector joins $L$ to the core, and the spokes attach the remaining terminal clones. Hence $D(G)$ is antistrong.

We next prove that $\UG(D(G))$ is 3-vertex-connected. Let $H=\UG(D(G))$.
By \eqref{eq:underlying-core}, the subgraph of $H$ induced by the core vertices is isomorphic to $K_7-02$.

Let $Z\subseteq V(H)$ with $|Z|\le 2$. If at least one of the vertices $0,1$ is not in $Z$, then every remaining terminal is adjacent to a remaining vertex in $\{0,1\}$. Moreover, the remaining core is connected, since deleting at most two vertices from $K_7-02$ leaves a connected graph. Hence $H-Z$ is connected.

The only remaining case is
$Z=\{0,1\}$. The terminal vertices then induce a copy of the connected graph $G$, the remaining core is connected, and the connector
$x_{u_*}2$ joins the terminal copy to the core. Thus $H-\{0,1\}$ is connected. We conclude that $H$ is 3-vertex-connected. In particular, $H$ is also 2-edge-connected.

Finally, only the core vertices $0$ and $1$ can have in-degree or out-degree growing with $|V(G)|$. For the terminals,
\[
\begin{aligned}
 \bigl(d_{D(G)}^+(x_u),d_{D(G)}^-(x_u)\bigr)
   &=(3,2)
   &&\text{for }u\in U\setminus\{u_*\},\\
 \bigl(d_{D(G)}^+(x_{u_*}),d_{D(G)}^-(x_{u_*})\bigr)
   &=(4,2),\\
 \bigl(d_{D(G)}^+(x_w),d_{D(G)}^-(x_w)\bigr)
   &=(2,3)
   &&\text{for }w\in W.
\end{aligned}
\]
For the remaining core vertices, direct inspection gives
\[
\begin{array}{c|ccccc}
 v & 2 & 3 & 4 & 5 & 6\\ \hline
 d_{D(G)}^+(v) & 3 & 4 & 3 & 3 & 2\\
 d_{D(G)}^-(v) & 3 & 2 & 3 & 3 & 4
\end{array}
\]
where the connector is included in the in-degree of $2$. Consequently, every vertex other than $0$ and $1$ has both in-degree and out-degree at most four.

The fixed no-instance $D_0$ also satisfies all these promises, as verified above.  Therefore the reduction proves NP-hardness under
the restrictions stated in Theorem~\ref{thm:main-2EC}.  Since the problem belongs to $\NP$, the proof is complete.
\end{proof}

\section{Concluding remarks}
Theorems~\ref{thm:main-SC} and~\ref{thm:main-2EC} show that the decision problems underlying Questions~9.1 and~9.2 of Bang-Jensen et al.~\cite{BBJK} are NP-complete.  Consequently,
both questions have negative answers unless $\mathsf{P}=\mathsf{NP}$.  The reductions also show that the difficulty is not caused by testing either connectivity property in isolation.  Antistrong connectivity is recognized by an ordinary connectivity test in the bipartite  representation, strong connectivity is recognized by strongly connected components, and 2-edge-connectivity is recognized by bridge detection. The hardness
lies in allocating a common arc set between two incompatible spanning requirements.

The first reduction already gives bounded in-degree and
out-degree. The second reduction gives a different strengthening:
the underlying graph is 3-vertex-connected, and only the two core vertices $0$ and $1$ can have unbounded in-degree or out-degree.
A natural remaining question is whether these two exceptional routing vertices can be replaced by bounded-degree gadgets, thereby
obtaining NP-completeness for a fully bounded-semidegree class while retaining the oriented, strong, antistrong and 3-vertex-connected promises.

\section*{Acknowledgements}

J.~Ai, H.~Lei and Y.~Shi were supported by the Fundamental and Interdisciplinary Disciplines Breakthrough Plan of the Ministry of
Education of China (JYB2025XDXM207).


\begin{thebibliography}{99}
\setlength{\itemsep}{0pt}

\bibitem{BBJK}
J.~Bang-Jensen, S.~Bessy, B.~Jackson and M.~Kriesell,
Antistrong digraphs,
\emph{J. Combin. Theory Ser. B} 122 (2017), 68--90.

\bibitem{BJGbook}
J.~Bang-Jensen and G.~Gutin,
\emph{Digraphs: Theory, Algorithms and Applications}, 2nd ed.,
Springer, London, 2009.

\bibitem{BJGY}
J.~Bang-Jensen, G.~Gutin and A.~Yeo,
Arc-disjoint strong spanning subdigraphs of semicomplete compositions,
\emph{J. Graph Theory} 95 (2020), 267--289.

\bibitem{BJKsurvey}
J.~Bang-Jensen and M.~Kriesell,
Disjoint sub(di)graphs in digraphs,
\emph{Electron. Notes Discrete Math.} 34 (2009), 179--183.

\bibitem{BJY}
J.~Bang-Jensen and A.~Yeo,
Decomposing $k$-arc-strong tournaments into strong spanning
subdigraphs,
\emph{Combinatorica} 24 (2004), 331--349.

\bibitem{BY}
J.~Bang-Jensen and A.~Yeo,
Arc-disjoint spanning sub(di)graphs in digraphs,
\emph{Theoret. Comput. Sci.} 438 (2012), 48--54.

\bibitem{Edmonds}
J.~Edmonds,
Edge-disjoint branchings,
in R.~Rustin (ed.), \emph{Combinatorial Algorithms},
Courant Computer Science Symposium 9,
Academic Press, New York, 1973, 91--96.


\bibitem{GJT}
M.~R.~Garey, D.~S.~Johnson and R.~E.~Tarjan,
The planar Hamiltonian circuit problem is NP-complete,
\emph{SIAM J. Comput.} 5 (1976), 704--714.

\bibitem{Munaro}
A.~Munaro,
On line graphs of subcubic triangle-free graphs,
\emph{Discrete Math.} 340 (2017), 1210--1226.

\bibitem{NashWilliams}
C.~St.~J.~A.~Nash-Williams,
Edge-disjoint spanning trees of finite graphs,
\emph{J. London Math. Soc.} 36 (1961), 445--450.

\bibitem{SGA}
Y.~Sun, G.~Gutin and J.~Ai,
Arc-disjoint strong spanning subdigraphs in compositions and products
of digraphs,
\emph{Discrete Math.} 342 (2019), 2297--2305.

\bibitem{Tutte}
W.~T.~Tutte,
On the problem of decomposing a graph into $n$ connected factors,
\emph{J. London Math. Soc.} 36 (1961), 221--230.

\end{thebibliography}
\end{document}